\documentclass{amsproc}

\usepackage{amsthm,amsmath,amsfonts,amssymb,amsxtra,appendix,bookmark,dsfont,bm,color}
\usepackage[shortlabels]{enumitem}

\newtheorem{theorem}{Theorem}[section]
\newtheorem{lemma}[theorem]{Lemma}

\theoremstyle{definition}

\theoremstyle{remark}

\numberwithin{equation}{section}

\newcommand{\abs}[1]{\lvert#1\rvert}

\newcounter{step} 

\usepackage{etoolbox}
\AtBeginEnvironment{proof}{\setcounter{step}{0}}

\DeclareMathOperator{\supp}{supp}
\DeclareMathOperator{\Tr}{Tr}

\def\bq{\begin{eqnarray}}
\def\eq{\end{eqnarray}}
\def\bqq{\begin{align*}}
\def\eqq{\end{align*}}

\def\nn{\nonumber}

\def\eps{\varepsilon}

\newcommand{\norm}[1]{\left\lVert #1 \right\rVert}

\newcommand\1{{\ensuremath {\mathds 1} }}

\newcommand*\dotv{{}\cdot{}}

\def\R {\mathbb{R}}

\def\cE {\mathcal{E}}

\def\cR{\mathcal{R}}

\def\R {\mathbb{R}}

\def\d{{\, \rm d}}

\begin{document}

\title[A proof of the ionization conjecture in M\"uller theory]{A short proof of the ionization conjecture in M\"uller theory}

\author{Rupert L. Frank}
\address[R. L. Frank]{Mathematisches Institut, Ludwig-Maximilans Universit\"at M\"unchen, Theresienstr. 39, 80333 M\"unchen, Germany, and Mathematics 253-37, Caltech, Pasadena, CA 91125, USA}
\email{rlfrank@caltech.edu}
\thanks{The first author was supported in part by U.S. NSF Grant DMS-1363432.}

\author{Phan Th\`anh Nam}
\address[P.T. Nam]{Department of Mathematics and Statistics, Masaryk University, Kotl\'a\v rsk\'a 2, 611 37 Brno, Czech Republic}
\email{ptnam@math.muni.cz}

\author{Hanne Van Den Bosch}
\address[H. Van Den Bosch]{Instituto de F\'{i}sica, Pontificia Universidad Cat\'olica de Chile, Av. Vicu\~na Mackenna 4860, Santiago, Chile}
\email{hannevdbosch@fis.puc.cl}
\thanks{The third author was supported in part by Conicyt (Chile) through CONICYT--PCHA/Docto\-ra\-do Nacional/2014 Project \# 116--0856 and Iniciativa Cient\'i\-fica Milenio (Chile) through Millenium Nucleus RC--120002 ``F\'isica Matem\'a\-tica''.}




\begin{abstract}
We prove that in M\"uller theory, a nucleus of charge $Z$ can bind at most $Z+C$ electrons for a constant $C$ independent of $Z$. 
\end{abstract}

\maketitle



\section{Introduction}

In M\"uller theory \cite{Mueller-84}, the energy of an atom is given by the functional
$$
\cE^{\rm M}(\gamma) = \Tr (-\Delta \gamma) - \int_{\R^3} \frac{Z\rho_\gamma(x)}{|x|} \d x + D(\rho_\gamma) - X(\gamma^{1/2}).
$$
Here $\gamma$ is the density matrix of the electrons and $\rho_\gamma(x)=\gamma(x,x)$ is its density. The Coulomb repulsion between the electrons is modeled by 
$$
D(\rho_\gamma)  = \frac{1}{2} \iint_{\R^3\times \R^3} \frac{\rho_\gamma(x) \rho_\gamma(y) }{|x-y|} \d x \d y
$$
and the exchange energy is described by
$$
X(\gamma^{1/2}) = \frac{1}{2} \iint_{\R^3\times \R^3} \frac{|\gamma^{1/2}(x,y)|^2 }{|x-y|} \d x \d y.
$$
The ground state energy is then given by
\begin{equation}
\label{eq:muellermin}
E^{\rm M}(N)= \inf \left\{ \cE^{\rm M}(\gamma)\,|\, 0\le \gamma\le 1\text{ on } L^2(\R^3), \Tr \gamma=N \right\}.
\end{equation}
Here we ignore the electron spin for the sake of simplicity. Moreover, for our mathematical treatment we do not need to assume that the parameters $Z>0$ (the nuclear charge) and $N>0$ (the number of electrons) are integers.

M\"uller theory is a modification of Hartree--Fock theory, where the usual exchange energy $X(\gamma)$ is replaced by $X(\gamma^{1/2})$. On one hand, like Hartree--Fock theory \cite{Bach-92}, M\"uller theory correctly reproduces the Scott and Dirac--Schwinger corrections to Thomas--Fermi theory; see \cite{Siedentop-09}. On the other hand, unlike the Hartree--Fock functional, the M\"uller functional is convex \cite{FraLieSieSei-07} and this leads to various mathematical simplifications. In particular, it follows from the discussion in \cite[Subsection I.C]{FraLieSieSei-07} that the density of any minimizer (if it exists) is radially symmetric.

In \cite{FraLieSieSei-07}, it was shown that the M\"uller functional has a minimizer if $N\le Z$, and it was conjectured that there is no minimizer if $N>N_c(Z)$ for a critical electron number $N_c(Z)<\infty$. As pointed out in \cite{FraLieSieSei-07}, in M\"uller theory some electrons may form a nontrivial bound state at infinity, and therefore it is unclear how to apply the standard method of ``multiplying the Euler-Lagrange equation by $|x|$" by Benguria and Lieb \cite{Benguria-79,Lieb-81b,Lieb-84}.

In \cite{FraNamBos-16b}, we used a different method to justify this conjecture and proved

\begin{theorem}\label{main}
There is a constant $C>0$ such that for all $Z>0$, the M\"uller variational problem \eqref{eq:muellermin} has no miminizer if $N> Z+C$.
\end{theorem}

The proof of Theorem \ref{main} in \cite{FraNamBos-16b} is adapted from our previous work on Thomas--Fermi--Dirac--von Weizs\"acker theory \cite{FraNamBos-16}. It consists of two main ingredients. The first one is a new strategy to control the number of electrons far away from the nucleus, which is inspired by \cite{NamBos-16} and \cite{FraKilNam-16}. The second one is a comparison with Thomas--Fermi theory, following Solovej's fundamental work on Hartree--Fock theory \cite{Solovej-03}. In \cite{FraNamBos-16b}, we did not use the convexity of M\"uller funtional in order to illustrate the generality of our strategy. In fact, our proof has been generalized in \cite{Kehle-16} to cover a class of non-convex models between M\"uller and Hartree--Fock. 

In this short note, we will provide a shorter proof of Theorem \ref{main} by using the convexity of M\"uller functional and following Solovej's proof in reduced Hartree--Fock theory \cite{Solovej-91}. 

\subsection*{Acknowledgement}
The first and second author are grateful to the organizers of the QMath 13 conference and for the invitation to speak there.


\section{Exterior $L^1$-estimate}

Throughout the paper we will assume that $N\ge Z$ and that the variational problem $E^{\rm M}(N)$ has a minimizer $\gamma_0$. As mentioned before we know that the density $\rho_0=\rho_{\gamma_0}$ is radially symmetric. In many places we will use Newton's theorem
$$
\int_{|y|<|x|} \frac{\rho_0(y)}{|x-y|} \d y = \frac{1}{|x|} \int_{|y|<|x|} \rho_0(y) \d y.
$$

We start by proving a simple bound, which in particular verifies the conjecture in \cite{FraLieSieSei-07} that there is a critical electron number $N_c(Z)<\infty$.  

\begin{lemma} \label{lem:2Z} $ N \le 2Z + C(Z^{2/3}+1).$
\end{lemma}

\begin{proof} For any partition of unity $\chi_1^2+\chi_2^2=1$, we have the binding inequality
\bq \label{eq:binding-inequality}
\cE^{\rm M}(\gamma_0) \le \cE^{\rm M}(\chi_1 \gamma_0 \chi_1) + \cE^{\rm M}_{Z=0}(\chi_2 \gamma_0 \chi_2).
\eq
We choose
$$
\chi_j (x) = g_j \Big( \frac{\nu \cdot x -\ell}{s}\Big)
$$ 
with $s>0,\ell>0, \nu\in \mathbb{S}^2$, and $g_j: \R \to \R^+$ satisfying 
$$
g_1^2+g_2^2=1, \quad g_1(t)=1 \text{~if~} t \le 0, \quad g_1(t)=0 \text{~if~} t \ge 1,\quad |\nabla g_1| + |\nabla g_2| \le C.
$$
By the IMS formula and the fact that
$$ X(\chi_j \gamma_0^{1/2} \chi_j) \le X((\chi_j \gamma_0 \chi_j)^{1/2}) $$
(see \cite[Lemma 3]{FraLieSieSei-07}), we can estimate
\begin{align*}
&\cE^{\rm M}(\chi_1 \gamma_0 \chi_1) + \cE^{\rm M}_{Z=0}(\chi_2 \gamma_0 \chi_2) - \cE^{\rm M}(\gamma_0) \\
&\le \int \Big( |\nabla \chi_1(x)|^2 + |\nabla \chi_2(x)|^2\Big) \rho_0(x) \d x + \int \frac{Z\chi_2^2(x) \rho_0(x)}{|x|} \d x \\
& \quad + \iint \frac{\chi_2^2(x) \Big( |\gamma_0^{1/2}(x,y)|^2 - \rho_0(x)\rho_0(y)\Big) \chi_1^2(y) }{|x-y|} \d x \d y \\
&\le Cs^{-2} \int_{\nu \cdot x -s \le \ell \le \nu \cdot x } \rho_0(x) \d x + \int_{\ell \le \nu \cdot x } \frac{Z \rho_0(x)}{|x|} \d x \\
& \quad +\iint_{\nu \cdot y - s \le \ell \le \nu \cdot x} \frac{  |\gamma_0^{1/2}(x,y)|^2}{|x-y|} \d x \d y- \iint_{\nu \cdot y \le \ell \le \nu \cdot x-s} \frac{ \rho_0(x) \rho_0(y)}{|x-y|} \d x \d y .
\end{align*}
Thus  from \eqref{eq:binding-inequality} it follows that for all $s>0,\ell>0$ and $\nu \in \mathbb{S}^2$,
\begin{align*} 
&\iint_{\nu \cdot y \le \ell \le \nu \cdot x-s} \frac{ \rho_0(x) \rho_0(y)}{|x-y|} \d x \d y \le  Cs^{-2} \int_{\nu \cdot x -s \le \ell \le \nu \cdot x } \rho_0(x) \d x  \nn\\
&\qquad + \int_{\ell \le \nu \cdot x } \frac{Z \rho_0(x)}{|x|} \d x  + \iint_{\nu \cdot y - s \le \ell \le \nu \cdot x} \frac{  |\gamma^{1/2}(x,y)|^2}{|x-y|} \d x \d y.
\end{align*}
Next, we integrate over $\ell \in (0,\infty)$, then average over $\nu \in \mathbb{S}^2$. We use Fubini's theorem and  
$$
\int_{\mathbb{S}^2} [\nu\cdot z]_+\,\frac{d\nu}{4\pi} =  \frac{|z|}{4},  \quad \forall z\in \mathbb{R}^3.
$$
Moreover, we also use 
$$
\int_0^\infty \1\big(b -s  \le \ell \le a\big)  \d \ell \le [a-b]_+ +s 
$$
(for the right side) and
$$
\int_0^\infty  \Big( \1\big(b \le \ell \le a - s\big) + \1\big(- a \le \ell \le - b - s\big) \Big) \d \ell \ge \Big[ [a-b]_+ -2s \Big]_+
$$
(for the left side) with $a=\nu \cdot  x$, $b=\nu \cdot y$  . All this leads to 
$$
N^2 \le Cs^{-1}N + 2Z N   + 2N + 2s ( D(\rho_0)+ X(\gamma_0^{1/2})).  
$$
Optimizing over $s>0$ and using the a-priori estimate 
$$D(\gamma_0)+X(\gamma_0^{1/2}) \le C(Z^{7/3}+N)$$
(which follows by an easy energy comparison; see \cite[Corollary 5]{FraNamBos-16b}), we get $ N \le 2Z + C(Z^{2/3}+1)$. 
\end{proof}

In order to improve the bound in Lemma \ref{lem:2Z}, we use the following observation. Heuristically, the electrons in the exterior region $|x|\ge r$ feel the rest of the system as an ``effective nucleus"  with the {\em screened nuclear charge}
\begin{equation}
\label{eq:screenedcharge}
Z_r = Z - \int_{|x|<r} \rho_0(x) \d x.
\end{equation}
Therefore, by modifying the proof of Lemma \ref{lem:2Z} we can control the number of exterior electrons in terms of $Z_r$. We still lose a factor 2, but this is not a big problem because $Z_r$ is much smaller than $Z$ (if $r$ is not too small).  

Throughout the paper, we will use the cut-off functions
\bq \label{eq:def-eta-r} \chi_r^+ (x)=\1(|x|\ge r), \quad \chi_r^+ \ge \eta_r \ge \chi_{(1+\lambda)r}^+, \quad \abs{\nabla \eta_r} \le C (\lambda r)^{-1}.
\eq
We have the following upgraded version of Lemma \ref{lem:2Z}.

\begin{lemma}[Exterior $L^1$-estimate] \label{lem:L1-bound} For all $r>0,s>0$ and $\lambda\in (0,1/2]$,
\begin{align*} 
 \int \chi^+_{r}\rho_0  &\le C  \int_{r<|x|<(1+\lambda)^2 r} \rho_0 + C \Big([Z_r]_+  + s + \lambda^{-2}s^{-1}+ \lambda^{-1}\Big)  \\
 &\quad +C  \Big( s^2 \Tr (-\Delta \eta_r \gamma_0 \eta_r) \Big)^{3/5} + C \Big( s^2 \Tr (-\Delta \eta_r \gamma_0 \eta_r) \Big)^{1/3}.
\end{align*} 
\end{lemma}
\begin{proof} We use the binding inequality \eqref{eq:binding-inequality} with 
$$\chi_j (x) = g_j \Big( \frac{\nu \cdot \theta(x) -\ell}{s}\Big)$$
where $\theta: \R^3\to \R^3$ satisfies 
$$
|\theta (x)| \le |x|, \quad \theta(x) =0 \text{~if~} |x| \le r, \quad \theta(x)=x \text{~if~} |x| \ge (1+\lambda) r, \quad |\nabla \theta|\le C \lambda^{-1}
$$
and proceed similarly as in Lemma \ref{lem:2Z}. See \cite[Lemma 7]{FraNamBos-16b} for details. 
\end{proof}


\section{Comparison with Thomas--Fermi theory}

In this section, we control the electron density in the exterior region $\{|x|\ge r\}$ in M\"uller theory by comparison with Thomas--Fermi (TF) theory. Recall that in usual TF theory, the ground state energy is obtained by minimizing the density functional 
$$
\cE^{\rm TF}(\rho)= c^{\rm TF}\int_{\R^3} \rho^{5/3}(x) \d x - \int_{\R^3} \frac{Z \rho(x)}{|x|} \d x + D(\rho), \quad c^{\rm TF}= \frac{3}{5}(6\pi^2)^{2/3},
$$
over all $0\le \rho \in L^1(\R^3)\cap L^{5/3}(\R^3)$. The TF minimizer $\rho^{\rm TF}$ is unique and has total mass $\int \rho^{\rm TF}=Z$ \cite{LieSim-77b}. Here, as in \cite{Solovej-91}, we will consider TF theory restricted to the exterior region $\{|x|\ge r\}$. 

\begin{lemma}[Exterior TF theory] \label{lem:TF} Let $r>0$ and $z \in \mathbb{R}$. Then the TF functional
$$
\cE_r^{\rm TF}(\rho)= c^{\rm TF} \int_{\R^3} \rho(x)^{5/3}\d x - \int_{\R^3} \frac{z \rho(x)}{|x|} \d x +D(\rho)$$
has a unique minimizer $\rho_r^{\rm TF}$ among all densities satisfying
$$0\le \rho \in L^{5/3}(\R^3)\cap L^1(\R^3), \quad \supp \rho \subset \{|x| \ge r\}.$$
The minimizer $\rho_r^{\rm TF}$ is radially symmetric, has total mass $\int \rho_r^{\rm TF}=[z]_+$, has bounded kinetic energy
\begin{equation}
\label{eq:rhotfrkin}
\int \left(\rho^{\rm TF}_r\right)^{5/3} \leq C [z]_+^{7/3}
\end{equation}
and satisfies the TF equation
\bq \label{eq:TFequation}
\frac{5c^{\rm TF}}{3}(\rho_r^{\rm TF})^{2/3} = [\varphi_r^{\rm TF}]_+
\qquad\text{in}\ \{|x|>r\}
\eq
with
$$
\varphi_r^{\rm TF}(x) = \frac{z \chi_r^+(x)}{|x|} - \rho_r^{\rm TF}*|x|^{-1}.
$$
Moreover, for every fixed $\kappa>0$, there is an $\alpha(\kappa)>0$ such that if $z r^3 \ge \kappa$ and $|x|r^{-1}\ge \alpha(\kappa)$, then we have the Sommerfeld estimate
\bq \label{eq:Sommerfeld}
\left| \rho_r^{\rm TF}(x) - \big(5\pi^{-1}c^{\rm TF}\big)^3 |x|^{-6} \right| \le  C |x|^{-6}(|x|r^{-1})^{-\zeta}
\eq
with $\zeta=(\sqrt{73}-7)/2\approx 0.77.$ For the full $\rho^{\rm TF}$, for all $x\ne 0$ we have
\bq \label{eq:Sommerfeld-r0}
0 \ge  \rho^{\rm TF}(x)  - \big(5\pi^{-1}c^{\rm TF}\big)^3 |x|^{-6} \ge - C |x|^{-6} (|x|z^{1/3})^{-\zeta}.
\eq
\end{lemma}

\begin{proof} See \cite[Appendix B]{Solovej-91}. In fact, \eqref{eq:Sommerfeld} is slightly stronger than \cite[Theorem B3]{Solovej-91} and it is taken from \cite[Lemma 4.4]{Solovej-03}. The bound \eqref{eq:Sommerfeld-r0} is taken from \cite[Theorems 5.2, 5.4]{Solovej-03}. 
\end{proof}

\emph{Convention.} In what follows, $\rho_r^{\rm TF}$ denotes the minimizer from Lemma \ref{lem:TF} with the choice $z=Z_r$ from \eqref{eq:screenedcharge}.

\medskip

The main result in this section is the following

\begin{lemma}[Comparison with TF] \label{lem:D-M-TF} For all $r\ge s > 0$ and $\lambda\in (0,1/2]$, 
\begin{align} \label{eq:D-D-D}
D(\eta_r^2 \rho_0- \rho_r^{\rm TF}) &\le  C s^{-2} \int \chi_r^+   \rho_0 + C [Z_r]_+^{12/5} r^{-1/5} s^{2/5} 
  + \cR \end{align}

where
\begin{align*}
\cR &= C \big(1+ (\lambda r)^{-2}\big) \int_{(1-\lambda)r \le |x| \le (1+\lambda)r} \rho_0 +  C \lambda r^{1/2} [Z_{(1-\lambda)r}]_+^{5/2} \nn\\
&\quad + C \Big( \Tr(-\Delta \eta_r \gamma_0 \eta_r) \Big)^{1/2} \Big( \int \chi_r^+ \rho_0 \Big)^{1/2}.
\end{align*}
\end{lemma}

To prove this lemma we will use the following semi-classical estimates from \cite[Lemma 8.2]{Solovej-03}.

\begin{lemma}[Semi-classical analysis]\label{lem:semi} Let $L_{\rm sc}= (15 \pi^2)^{-1}$.
For every $s>0$, fix a smooth function $g_s:\R^3\to [0,\infty)$ such that
$$
\supp g_s \subset \{|x| \le s\}, \quad \int g_s^2=1,\quad \int |\nabla g_s|^2 \le Cs^{-2}. 
$$
(i) For all $V: \R^3\to \R$ such that $[V]_+, [V-V*g_s^2]_+ \in L^{5/2}$ and for all density matrices $0\le \gamma \le 1$, we have
\begin{align} \label{eq:semi-lower}
\Tr( (-\Delta-V)\gamma) &\ge - L_{\rm sc} \int [V]_+^{5/2} - C s^{-2} \Tr \gamma \nn\\
&\qquad - C \left( \int [V]_+^{5/2}\right)^{3/5}\left( \int [V-V*g_s^2]_+^{5/2}\right)^{2/5}.
\end{align}
(ii) If $V_+\in L^{5/2}(\R^3)\cap L^{3/2}(\R^3)$, then there is a density matrix $\gamma$ such that
$$
\rho_\gamma= \frac{5}{2}L_{\rm sc} [V]_+^{3/2}*g_s^2
$$
and
\begin{align} \label{eq:semi-upper}
\Tr( -\Delta\gamma) \leq \frac{3}{2}L_{\rm sc}\int [V]_+^{5/2} + C  s^{-2} \int [V]_+^{3/2}.
\end{align}
\end{lemma}


\begin{proof}[Proof of Lemma \ref{lem:D-M-TF}]
\noindent {\bf Step 1.} 
First, we show that the exterior density matrix $\eta_r \gamma_0 \eta_r$ essentially minimizes the exterior reduced Hartree-Fock functional 
$$\cE^{\rm RHF}_r(\gamma)= \Tr(-\Delta \gamma) -\int_{\R^3} \frac{Z_r \rho_\gamma (x)}{|x|}  \d x + D(\rho_\gamma) \,,
$$
where $Z_r$ is given by \eqref{eq:screenedcharge}. Indeed, for all $r>0$, $\lambda\in (0,1/2]$ and for all density matrices 
$$ 0\le \gamma \le 1, \quad {\rm supp}(\rho_\gamma) \subset \{|x| \ge r\}, \quad \Tr \gamma \le \int \chi_r^+\rho_0, $$
we have 
\begin{align} \label{eq:rHF}
\cE^{\rm RHF}_r(\eta_r \gamma_0 \eta_r)  \le \cE^{\rm RHF}_r(\gamma)+ \cR
\end{align}
The proof of \eqref{eq:rHF} is straightforward, using a trial state argument. We refer to  \cite[Lemma 9]{FraNamBos-16b} for details.

\medskip

\noindent {\bf Step 2.} Now we bound the right side of \eqref{eq:rHF} by choosing $\gamma$ as in Lemma~\ref{lem:semi} (ii) with $V = \chi_{r+s}^+\varphi_r^{\rm TF} \ge 0$. Note that $\rho_{\gamma} = \left(\chi_{r+s}^+ \rho_r^{\rm TF}\right)* g_s^2$ by the TF equation, and hence 
$$
\supp \rho_{\gamma}\subset \{|x|\ge r\},\quad \Tr\gamma = \int \rho_{\gamma} = \int \chi_{r+s}^+ \rho_r^{\rm TF} \leq \int \rho_r^{\rm TF} =Z_r \leq \int \chi_r^+\rho_0.
$$
The last inequality here comes from our assumption $N =\int \rho_0 \geq Z$. On the other hand, by the semi-classical estimate \eqref{eq:semi-upper},
\begin{align} \label{eq:rHF-upper}
\cE_r^{\rm RHF} ( \gamma) & \le 
        \frac{3}{2}L_{\rm sc}\int [\varphi_r^{\rm TF}]^{5/2}  - \int \frac{Z_r}{|x|} \rho_\gamma + D(\rho_\gamma) + C  s^{-2} \int [\varphi_r^{\rm TF}]^{3/2}_+ \nn\\
              &\le \cE^{\rm TF}_r(\rho_r^{\rm TF})   + \int_{r \le \abs{x} \le r+s}  \frac{Z_r}{|x|}\rho_r^{\rm TF} + C  s^{-2} \int \rho_r^{\rm TF}.
\end{align}
Here we have used the TF equation, the convexity $
D(\rho_\gamma)\le D(\rho_r^{\rm TF}* g_s^2) \le D(\rho_r^{\rm TF})$ and Newton's theorem
$$
\int |x|^{-1}\rho_\gamma  = \int   (|x|^{-1}*g_s^2) (\chi_{r+s}^+\rho_r^{\rm TF}) = \int |x|^{-1} (\chi_{r+s}^+\rho_r^{\rm TF}).
$$
Finally, we bound the error term by using \eqref{eq:rhotfrkin} and H\"older's inequality:
\begin{align} \label{eq:rho/x-by-kinetic}
\int_{r\le |x|\le r+s} \frac{Z_r}{|x|}\rho_r^{\rm TF} &\leq C \frac{[Z_r]_+}{r} \Big( \int_{r\le |x|\le r+s} (\rho_r^{\rm TF}(x) )^{5/3} \d x \Big)^{3/5} \Big(\int_{r\le |x|\le r+s} \d x \Big)^{2/5} \nn \\
& \le C \frac{[Z_r]_+}{r} \Big( [Z_r]_+^{7/3} \Big)^{3/5} \Big( r^2 s \Big)^{2/5} = C [Z_r]_+^{12/5} r^{-1/5} s^{2/5}.
\end{align}

\noindent {\bf Step 3.} To bound the left side of \eqref{eq:rHF}, we write
\begin{align*} 
 \cE_r^{\rm RHF} (\eta_r \gamma_0 \eta_r) 
    &= \Tr ((-   \Delta - \varphi_r^{\rm TF})\eta_r \gamma_0 \eta_r)  - D(\rho_r^{\rm TF}) + D(\eta_r^2 \rho_0- \rho_r^{\rm TF}) 
\end{align*}
and use the semi-classical estimate \eqref{eq:semi-lower} with $V= \varphi_r^{\rm TF}$. Note that by Newton's theorem,
\begin{align*}
 - \rho_r^{\rm TF}* (|x|^{-1} - |x|^{-1}* g_s^2) \le 0
\end{align*}
and
\begin{align*}
 \big[ \chi_r^+|\dotv|^{-1} - (\chi_r^+|\dotv|^{-1})* g_s^2 \big]_+(x) &\le  \big[ \chi_r^+|\dotv|^{-1} - \chi_{r+s}^+ (|\dotv|^{-1}* g_s^2) \big]_+(x)  \nn\\ 
 & =  (\chi_r^+(x) - \chi_{r+s}^+ (x))|x|^{-1}. 
\end{align*}
Therefore, when $Z_r \ge 0$ and $r\ge s$  we can bound
 \begin{align}\label{eq:V*g-V}
\left[ \varphi_r^{\rm TF} - \varphi_r^{\rm TF}*g_s^2 \right]_+(x) \le [Z_r]_+ (\chi_r^+(x) - \chi_{r+s}^+ (x))|x|^{-1}.
\end{align} 
%
%
Using the TF equation \eqref{eq:TFequation} and the TF kinetic energy bound \eqref{eq:rhotfrkin}, we get, similarly to \eqref{eq:rho/x-by-kinetic}, 
\begin{align} \label{eq:varphi-by-kinetic-b}
 \left( \int [\varphi_r^{\rm TF}]_+^{5/2}\right)^{3/5}\left( \int [\varphi_r^{\rm TF}-\varphi_r^{\rm TF}*g_s^2]_+^{5/2}\right)^{2/5} 
 &\le C \norm{\rho_r^{\rm TF}}_{L^{5/3}}[Z_r]_+  r^{4/5-1} s^{2/5} \nn \\
 & \le C [Z_r]_+^{12/5} r^{-1/5} s^{2/5} . 
\end{align}
Note that \eqref{eq:varphi-by-kinetic-b} holds independently of the sign of $Z_r$ since $[\varphi_r^{\rm TF}]_+=0$ if $Z_r \le 0$.  
Thus,
\begin{align} \label{eq:rHF-lower}
 \cE_r^{\rm RHF} (\eta_r \gamma_0 \eta_r) 
    &= \Tr ((-   \Delta - \varphi_r^{\rm TF})\eta_r \gamma_0 \eta_r)  - D(\rho_r^{\rm TF}) + D(\eta_r^2 \rho_0- \rho_r^{\rm TF})\nn\\
    & \ge - L_{\rm sc } \int[\varphi_r^{\rm TF}]^{5/2} - C s^{-2} \int \eta_r^2 \rho_0 - C [Z_r]_+^{12/5} s^{2/5} r^{-1/5}\nn \\
    & \qquad - D(\rho_r^{\rm TF}) + D(\eta_r^2 \rho_0- \rho_r^{\rm TF}) \nn \\
    & =  \cE^{\rm TF}_r(\rho_r^{\rm TF})   +  D(\eta_r^2 \rho_0- \rho_r^{\rm TF}) \nn \\
    & \qquad - C s^{-2} \int \eta_r^2 \rho_0 - C [Z_r]_+^{12/5} s^{2/5} r^{-1/5}.    
\end{align}
Putting together \eqref{eq:rHF}, \eqref{eq:rHF-upper}, \eqref{eq:rho/x-by-kinetic} and \eqref{eq:rHF-lower}, we obtain \eqref{eq:D-D-D}. 
\end{proof}

In order to translate \eqref{eq:D-D-D} into an $L^1$-estimate, we will need 

\begin{lemma} \label{lem:f*1/|x|} For every $f\in L^{5/3}(\R^3) \cap L^{1} (\R^3)$ and $x\in \R^3$, we have
$$
\left| \int_{|y|<r} f(y)  \d y \right| \le C  \|f\|_{L^{5/3}}^{5/6} D(f)^{1/12} r^{13/12}.
$$
\end{lemma}

\begin{proof} From \cite[Cor. 9.3]{Solovej-03} (see also \cite[Lem. 18]{FraNamBos-16}) we have
$$
\left| \int_{|y|<|x|} \frac{f(y)}{|x-y|} \d y \right| \le C \|f\|_{L^{5/3}}^{5/6} (|x|D(f))^{1/12}, \quad \forall x\in \R^3.
$$
Choosing $x=r\nu$ and averaging over $\nu \in \mathbb{S}^2$, we get the conclusion. 
\end{proof}

We finish this section by proving some a-priori estimates for $\chi_r^+\rho_0$.

\begin{lemma}[A-priori estimates] \label{lem:a-priori} Assume that 
\bq \label{eq:int-rho-r-r/2}
|Z_r|\le Cr^{-3}, \quad \forall r \in (0,D]
\eq
for some $D\le 1$. Then
\begin{align}
\int \chi_r^+ \rho_0 &\le C r^{-3},\quad  \forall r\in (0,D],\label{eq:int-rho-4}\\
\int \chi_r^+ \rho_0^{5/3} &\le C r^{-7}, \quad  \forall r\in (0,D],\label{eq:int-rho-3}\\
\Tr(-\Delta \eta_r \gamma_0 \eta_r) & \le C r^{-7}, \quad \forall r \in (0,D].\label{eq:Tr-eta-int-chir}
\end{align}
Here the constants are independent of $D$ and the cut-off function $\eta_r$ satisfies \eqref{eq:def-eta-r} with $\lambda\in [r/2,1/2]$.
\end{lemma}

\begin{proof}
We can choose $\gamma=0$ in \eqref{eq:rHF} to get $\cE_r^{\rm RHF}(\eta_r \gamma_0 \eta_r)\le \cR.$ Using the kinetic Lieb--Thirring inequality and TF lower bound, we find that
\begin{align*}
\Tr(-\Delta \eta_r \gamma_0 \eta_r) & \le C  (\lambda r)^{-2} \Big(Z_{(1-\lambda)r} - Z_{r} +\int \chi_r^+ \rho_0 \Big)  +  C \lambda r^{1/2} [Z_{(1-\lambda)r}]_+^{5/2}  \nn\\
& \quad + C \Big( \Tr(-\Delta \eta_r \gamma_0 \eta_r) \Big)^{1/2} \Big( \int \chi_r^+ \rho_0 \Big)^{1/2} +  C [Z_r]_+^{7/3}.
\end{align*}
This bound, \eqref{eq:int-rho-r-r/2} and the choice $\lambda\ge r/2$ imply that
\bq
\Tr(-\Delta \eta_r \gamma_0 \eta_r) \le C \Big(r^{-4} \int \chi_r^+ \rho_0 + r^{-7} \Big), \quad \forall r \in (0,D].\label{eq:Tr-eta-int-chir1}
\eq
We recall that the estimate in Lemma~\ref{lem:L1-bound} with $\lambda \ge r/2=s/2$ gives
\begin{align*} 
 \int \chi^+_{r}\rho_0  &\le C  \int_{r<|x|<(1+\lambda)^2 r} \rho_0 + C \Big([Z_r]_+  +  r^{-3}\Big)  \\
 &\quad +C  \Big( r^2 \Tr (-\Delta \eta_r \gamma_0 \eta_r) \Big)^{3/5} + C \Big( r^2 \Tr (-\Delta \eta_r \gamma_0 \eta_r)\Big)^{1/3}.
\end{align*} 
By inserting \eqref{eq:int-rho-r-r/2} and \eqref{eq:Tr-eta-int-chir1} into this bound we deduce that  
\begin{align*} 
 \int \chi^+_{r}\rho_0  \le C  \big(Z_r - Z_{(1+\lambda)^2 r} \big)+ C r^{-3}, \quad \forall r \in (0,D].
\end{align*} 
We can replace $r$ by $(1+\lambda)^{-2}r$ and  use \eqref{eq:int-rho-r-r/2} to get \eqref{eq:int-rho-4}. Inserting \eqref{eq:int-rho-4} into \eqref{eq:Tr-eta-int-chir1} yields \eqref{eq:Tr-eta-int-chir}. Moreover, by \eqref{eq:Tr-eta-int-chir}  and  the kinetic Lieb--Thirring inequality, we have
\[
\int \chi_{2r}^+ \rho_0^{5/3} \le \int (\eta_{r}^2 \rho_0)^{5/3} \le C \Tr(-\Delta \eta_{r} \gamma_0 \eta_{r}) \le C r^{-7}, \quad \forall r \in (0,D].
\]
Replacing $r$ by $r/2$ we obtain \eqref{eq:int-rho-3}.
\end{proof}

\section{Proof of the main result}

Now we prove Theorem \ref{main}. Since the (usual) TF minimizer $\rho^{\rm TF}$ has total mass $Z$ \cite{LieSim-77b}, Theorem \ref{main} is a direct consequence of the following

\begin{theorem}[Comparison with TF density] \label{thm:comparison} There are universal constants $C>0, \eps>0$ such that for all $N\ge Z \ge 1$ and $r>0$,
\bq \label{eq:comparison}\left| \int_{|x|<r}\Big( \rho_0(x) - \rho^{\rm TF}(x)\Big) \d x \right|  \le C(1+ r^{-3+\eps}).
\eq
\end{theorem}
Note that the left side of \eqref{eq:comparison} is $|Z_r-Z_r^{\rm TF}|$ where 
$$
Z_r^{\rm TF}:= Z - \int_{|x|<r} \rho^{\rm TF}(x) \d x = \int_{|x|\ge r} \rho^{\rm TF}(x) \d x. 
$$
Recall that by the Sommerfeld estimate  \eqref{eq:Sommerfeld-r0}, for all $r>0$ we have
\bq \label{eq:ZrTF} Z_r^{\rm TF} = a^{\rm TF}  r^{-3}\Big(1+O\big((rZ^{1/3})^{-\zeta}\big)\Big), \quad a^{\rm TF}:=\frac{4 (5c^{\rm TF})^3}{3\pi^2}.
\eq
Thus, \eqref{eq:comparison} tells us that the screened nuclear charge $Z_r$ can be approximated well by TF theory up to the distance $o(1)$, which is remarkably larger than the semi-classical distance $O(Z^{-1/3})$. 

We will prove Theorem \ref{thm:comparison} using a bootstrap argument as in \cite{Solovej-91}.
\begin{lemma}[Initial step] \label{lem:initial} There is a universal constant $C_1>0$ such that  
\bq \label{eq:comparison-1}
\left| \int_{|x|<r}\big( \rho_0(x) - \rho^{\rm TF}(x)\big) \d x \right|  \le C_1 (Z^{1/3}r)^{179/44} r^{-3+1/66}, \quad \forall r>0.
\eq
\end{lemma}
\begin{proof} By writing $\cE^{\rm M}(\gamma)=\cE^{\rm RHF}(\gamma)-X(\gamma^{1/2})$, we obtain 
$$
\cE^{\rm RHF}(\gamma_0) \le \inf \{ \cE^{\rm RHF}(\gamma): 0\le \gamma \le 1, \Tr \gamma \le N\}  + X(\gamma_0^{1/2}).
$$
Then we can use the semi-classical analysis as in the proof of Lemma \ref{lem:D-M-TF} (now with $r=0$ and $V(x) = \varphi^{\rm TF}(x) := Z|x|^{-1}-\rho^{\rm TF}*|x|^{-1} \ge 0$). The only difference is that instead of \eqref{eq:V*g-V} we use  
\begin{align*}
\int [\varphi^{\rm TF} - \varphi^{\rm TF}*g_s^2]_+^{5/2} \le Z^{5/2} \int \Big(|x|^{-1} - |x|^{-1}* g_s^2\Big)^{5/2} \le CZ^{5/2} s^{1/2}.
\end{align*}
We thus obtain 
$$
D(\rho_0- \rho^{\rm TF}) \le \int_{|x|\le s}  \frac{Z}{|x|}\rho^{\rm TF} + C s^{-2} N + C \| \rho^{\rm TF}\|_{L^{5/3}} Zs^{1/5} + X(\gamma_0^{1/2}).
$$
Using the a-priori estimates
$$N\le CZ, \quad \int (\rho^{\rm TF})^{5/3} \le CZ^{7/3}, \quad X(\gamma_0^{1/2}) \le CZ^{5/3} $$
and optimizing over $s>0$, we get $D(\rho_0- \rho^{\rm TF}) \le CZ^{25/11}.$ The desired estimate \eqref{eq:comparison-1} then follows from Lemma \ref{lem:f*1/|x|}. \end{proof}
\begin{lemma}[Iterative step] \label{thm:screened-it} There are universal constants $C_2, \delta, \eps>0$ such that, if for some $D\le 1$
\bq \label{eq:assume-D}
\left| \int_{|x|<r}\big( \rho_0(x) - \rho^{\rm TF}(x)\big) \d x \right|  \le (a^{\rm TF}/2) r^{-3}, \quad \forall r \in (0,D],
\eq
then  
\bq \label{eq:assume-D-it}
\left| \int_{|x|<r}\big( \rho_0(x) - \rho^{\rm TF}(x)\big) \d x \right| \le C_2 r^{-3+\eps}, \quad \forall r\in [D, D^{1-\delta}].
\eq
\end{lemma}

\begin{proof} Let $R\ge D\ge r$.
Since $\int  \rho_r^{\rm TF} = Z_r$ (see Lemma \ref{lem:TF}), we have the key identity
\begin{align} \label{eq:key-identity}
\int_{|x|<R} \left( \rho^{\rm TF} - \rho_0\right) & = Z_r - \int_{|x|< R} \chi_r^+ \rho_0 - \int_{|x|\ge R} \rho^{\rm TF} \nn\\
&=\int_{|x|< R} (\rho_r^{\rm TF} -\chi_r^+ \rho_0) + \int_{|x|\ge R} (\rho_r^{\rm TF} - \rho^{\rm TF}) .
\end{align}

In order to estimate the first term on the right hand side, we start by establishing \eqref{eq:int-rho-r-r/2}.
If $r\geq Z^{-1/3}$ we estimate $|Z_r| \leq |Z_r^{\rm TF}| + |Z_r-Z_r^{\rm TF}|$ and bound $\abs{Z_r^{\rm TF}}$ by \eqref{eq:ZrTF} (this requires $r\geq Z^{-1/3}$) and use the assumption \eqref{eq:assume-D}. On the other hand, if $r\leq Z^{-1/3}$, the bound $|Z_r|\leq C Z \leq C r^{-3}$ follows by our a-priori bound in Lemma \ref{lem:2Z}. This proves \eqref{eq:int-rho-r-r/2}.

Therefore, we can apply Lemma \ref{lem:a-priori} and we obtain the bounds \eqref{eq:int-rho-4}, \eqref{eq:int-rho-3} and \eqref{eq:Tr-eta-int-chir}. We want to use these in Lemma \ref{lem:D-M-TF} to bound $D(\eta_r^2 \rho_0- \rho_r^{\rm TF})$. First, we use Lemma \ref{lem:a-priori} to obtain 
\begin{align*}
\cR \le C(r^{-3} +\lambda^{-2} r^{-5} + \lambda r^{-7} + r^{-5}) \le C(\lambda^{-2} r^{-5} + \lambda r^{-7}) .
\end{align*}
For the other error terms in Lemma~\ref{lem:D-M-TF}, we choose $s= r^{11/6}$, so that
\[
C |Z_r|^{12/5} r^{-1/5} s^{2/5} + C s^{-2} \int \chi_r^+ \rho_0  \le C r^{-7+1/3}
\]
and finally obtain
\[
D(\eta_r^2 \rho_0- \rho_r^{\rm TF}) \le C r^{-7} \left( r^{1/3} +  \lambda^{-2} r^2 + \lambda \right).
\]

Inserting this bound, as well as \eqref{eq:int-rho-3}, \eqref{eq:rhotfrkin} and \eqref{eq:int-rho-r-r/2}, into Lemma \ref{lem:f*1/|x|}, we obtain 
$$
\left|\int_{|x|\le R} \Big(\eta_r^2 \rho_0- \rho_r^{\rm TF} \Big) \right|\le Cr^{-3} \left( r^{1/3} +  \lambda^{-2} r^2 + \lambda \right)^{1/12} (R/r)^{13/12}.
$$
Moreover, from \eqref{eq:int-rho-3}  and $0 \le \chi_r^+- \eta_r^2 \le \1(r\le |x| \le (1+\lambda)r)$ it follows that
$$ \int |\chi_r^+  \rho_0-\eta_r^2\rho_0| \le \|\chi_r^+ \rho_0\|_{L^{5/3}} \|\chi_r^+-\eta_r^2\|_{L^{5/2}} \le C r^{-3}\lambda^{2/5}.$$
Combining these estimates and choosing $\lambda=r^{1/3}$, we conclude that
\bq \label{eq:r<x<R}
\left|\int_{|x|\le R} \Big(\chi_r^+ \rho_0- \rho_r^{\rm TF} \Big) \right|\le Cr^{-3+1/36} (R/r)^{13/12}.
\eq

Next, we use the Sommerfeld asymptotics to bound $(\rho_r^{\rm TF}-\rho^{\rm TF})$ on  $\{|x|\ge R\}$. We will assume that $R\ge Lr \ge L^2Z^{-1/3}$ for a universal constant $L>0$ to be determined. From \eqref{eq:ZrTF}, by choosing $L>0$ large enough, we have
$$
Z^{\rm TF}_r \ge (3a^{\rm TF}/4)r^{-3}, \quad \forall r\ge L Z^{-1/3}.
$$
Combining this bound with \eqref{eq:assume-D} we infer that
$$
Z_r \ge (a^{\rm TF}/4)r^{-3}, \quad \forall r\ge L Z^{-1/3}.
$$
Because of this we can, after increasing $L$ if necessary, apply the Sommerfeld estimate \eqref{eq:Sommerfeld} and deduce that
$$
\left|\int_{|x|\ge R} \rho_r^{\rm TF}(x) \d x - a^{\rm TF} R^{-3} \right| \le CR^{-3} (R/r)^{-\zeta}, \quad \forall R \ge Lr.
$$
Combining the latter estimate and \eqref{eq:ZrTF} (with $r$ replaced by $R$), we finally obtain
\bq \label{eq:x>R-TF}
\left|\int_{|x|\ge R}  \Big(\rho_r^{\rm TF} -  \rho^{\rm TF} \Big) \right| \le CR^{-3} (R/r)^{-\zeta}, \quad \forall R \ge Lr.
\eq

Now let us conclude. From the bound \eqref{eq:comparison-1} in the initial step, by choosing universal constants $\delta>0$ and $\eps>0$ small enough we have
\begin{align*}
\left| \int_{|x|<r}\big( \rho_0 - \rho^{\rm TF}\big) \right|   \le Cr^{-3+\eps}, \quad \forall r\le (L^2Z^{-1/3})^{(1-\delta)^{2}}.
\end{align*}
Therefore, \eqref{eq:assume-D-it} holds true if $D \le (L^2Z^{-1/3})^{1-\delta}$. It remains to consider the case $(L^2Z^{-1/3})^{1-\delta}\le D \le 1$. We choose $r=L^{-1} D^{(1-\delta)^{-1}}$. Then we have $D\ge L r \ge L^2 Z^{-1/3}$. Therefore, for all $R\in [D,D^{1-\delta}]$, by inserting \eqref{eq:r<x<R} and \eqref{eq:x>R-TF} into \eqref{eq:key-identity}, we get
$$
\left|\int_{|x|< R}  \Big(\rho_0 -  \rho^{\rm TF} \Big) \right| \le  Cr^{-3+1/36} (R/r)^{13/12} + CR^{-3} (R/r)^{-\zeta}.
$$
Using $r\ge L^{-1}R^{(1-\delta)^{-2}}$ we  deduce that 
$$
\left|\int_{|x|< R}  \Big(\rho_0 -  \rho^{\rm TF} \Big) \right| \le  CR^{-3+\eps}, \quad \forall R\in [D,D^{1-\delta}]
$$
if the universal constants $\delta>0$ and $\eps>0$ are chosen small enough. 
\end{proof}

Now we are ready to provide the

\begin{proof}[Proof of Theorem \ref{thm:comparison}] We use the notations in  Lemma \ref{lem:initial} and Lemma \ref{thm:screened-it}. Let $\beta=a^{\rm TF}/2$ and $C_0=\max\{C_1,C_2,a^{\rm TF}\}$. The constant $\eps>0$ in Lemma \ref{thm:screened-it} can be chosen to satisfy $\eps\le 1/66$. Let $D_n:=Z^{-\frac{1}{3}(1-\delta)^n}$ for $n=0,1,2, \ldots $

From Lemma \ref{lem:initial}, we have
\bq \label{eq:initial-conclude}
\left| \int_{|x|<r}\big( \rho_0(x) - \rho^{\rm TF}(x)\big) \d x \right|  \le C_0 r^{-3+\eps}, \quad \forall r\in (0,D_0].
\eq
On the other hand, from Lemma \ref{thm:screened-it}, we deduce by induction that if 
$$C_0 (D_n)^\eps \le \beta,$$
then 
$$
\left| \int_{|x|<r}\big( \rho_0(x) - \rho^{\rm TF}(x)\big) \d x \right|  \le C_0 r^{-3+\eps}, \quad \forall r\in (0,D_{n+1}].
$$
Note that $D_n\to 1$ as $n\to\infty$ and that $C_0>\beta$. Thus, there is a minimal $n_0\in \{0,1,2,\ldots\}$ such that $C_0 (D_{n_0})^\eps>\beta$. If $n_0\ge 1$, then $C_0 (D_{n_0-1})^\eps\le \beta$ and therefore by the preceding argument
\bq \label{eq:Dn0}
\left| \int_{|x|<r}\big( \rho_0(x) - \rho^{\rm TF}(x)\big) \d x \right|  \le C_0 r^{-3+\eps}, \quad \forall r\in (0,D_{n_0}].
\eq
If $n_0=0$, then \eqref{eq:Dn0} reduces to \eqref{eq:initial-conclude}. Since $D_{n_0}\ge D:= (C_0^{-1} \beta)^{1/\eps}$, we have
\bq \label{eq:D}
\left| \int_{|x|<r}\big( \rho_0(x) - \rho^{\rm TF}(x)\big) \d x \right|  \le C_0 r^{-3+\eps}, \quad \forall r\in (0,D].
\eq

Note that $D\in (0,1]$ is a universal constant. Using the exterior estimates \eqref{eq:int-rho-4} and \eqref{eq:ZrTF} with $r=D$, we get
\bq \label{eq:D1}
\int_{|x|\ge D} \big( \rho_0(x)+ \rho^{\rm TF}(x) \big) \d x \le C.
\eq
From \eqref{eq:D} and \eqref{eq:D1}, we obtain \eqref{eq:comparison}. 
\end{proof}


\bibliographystyle{amsalpha}

\begin{thebibliography}{16}

\bibitem{Benguria-79} 
{\sc R. D. Benguria}, {\em The von Weizs\"acker and exchange corrections in the Thomas--Fermi theory}, Ph.D. Thesis, 
Princeton University, 1979.

\bibitem{Bach-92} {\sc V. Bach}, {\em Error bound for the Hartree-Fock energy of atoms and molecules}, Commun. Math. Phys. 147 (1992), pp. 527--548.

\bibitem{FraLieSieSei-07} {\sc R. L. Frank, E. H. Lieb, R. Seiringer and H. Siedentop}, {\em M\"uller's exchange-correlation energy in density-matrix-functional theory}. Phys. Rev. A 76 (2007), 052517.

\bibitem{FraKilNam-16} {\sc R. L. Frank, R. Killip and P. T. Nam}, {\em Nonexistence of large nuclei in the liquid drop model}, Lett. Math. Phys. 106 (2016), pp. 1033--1036. 
%
\bibitem{FraNamBos-16} {\sc R. L. Frank, P. T. Nam and H. Van Den Bosch}, {\em The ionization conjecture in Thomas-Fermi-Dirac-von Weizs\"acker theory}, Comm. Pure Appl. Math. 71 (2018), no. 3, pp. 577--614.

\bibitem{FraNamBos-16b} {\sc R. L. Frank, P. T. Nam and H. Van Den Bosch}, {\em The maximal excess charge in M\"uller density-matrix-functional theory}, Ann. H. Poincar\'e (2018), to appear. Preprint 2016, arXiv:1608.05625.

\bibitem{Kehle-16} {\sc C. Kehle}, {\em The maximal excess charge for a family of density-matrix-functional theories including Hartree--Fock and M\"uller theories}, J. Math. Phys. 58 (2017), no. 1, 011901.

\bibitem{Lieb-81b}
{\sc E.~H. Lieb}, {\em {Thomas--Fermi and related theories of atoms and molecules}}, Rev. Mod. Phys. 53 (1981), pp.~603--641.
%
\bibitem{Lieb-84}
\leavevmode\vrule height 2pt depth -1.6pt width 23pt, {\em Bound on the maximum negative ionization of atoms and molecules}, Phys. Rev. A 29 (1984),
  pp.~3018--3028.

\bibitem{LieSim-77b}
{\sc E.~H. Lieb and B.~Simon}, {\em The {T}homas--{F}ermi theory of atoms, molecules and solids}, Advances in Math., 23 (1977), pp.~22--116.

\bibitem{LieThi-75} {\sc E.~H. Lieb and W.~E. Thirring}. {\em Bound for the Kinetic Energy of Fermions Which Proves the Stability of Matter}, Phys. Rev. Lett. 35 (1975), 687. 

\bibitem{Mueller-84}
{\sc A. M. K. M\"uller}, {\em Explicit approximate relation between reduced two- and one-particle density matrices}, Phys. Lett. A 105 (1984), pp.~446--452

\bibitem{NamBos-16}
{\sc P.~T. Nam and H. Van Den Bosch}, {\em Nonexistence in Thomas-Fermi-Dirac-von Weizs\"acker theory with small nuclear charges},  Math. Phys. Anal. Geom. 20 (2017), no. 2, Art. 6.

\bibitem{Siedentop-09}
{\sc H. Siedentop}, {\em The asymptotic behaviour of the ground state energy of the M\"{u}ller functional for heavy atoms}, (German),  J. Phys. A: Math. Theor., 42 (2009), 085201.

\bibitem{Solovej-91}
{\sc J.~P. Solovej}, {\em {Proof of the ionization conjecture in a reduced Hartree-Fock model.}}, Invent. Math. 104 (1991), pp.~291--311.
%
\bibitem{Solovej-03}
\leavevmode\vrule height 2pt depth -1.6pt width 23pt, {\em The ionization
  conjecture in {H}artree-{F}ock theory}, Ann. of Math. (2) 158 (2003),
  pp.~509--576.
\end{thebibliography}

\end{document}